\documentclass[11pt,reqno]{amsart}
\usepackage[utf8]{inputenc}
\usepackage{setspace}

\usepackage{amsmath,amssymb,amsthm,mathrsfs,graphicx,caption,booktabs,multirow,adjustbox,makecell,float}
\usepackage[margin=3.4cm]{geometry}
\usepackage{amsthm}
\usepackage{enumitem}
\newtheoremstyle{tight}
{0pt}   
{0pt}   
{}      
{}      
{\bfseries} 
{.}     
{0.5em} 
{}      

\theoremstyle{tight}

\newtheorem{theorem}{Theorem}[section]
\newtheorem{example}{Example}[section]
\newtheorem{definition}{Definition}[section]
\newtheorem{corollary}{Corollary}[section]
\newtheorem{property}{Property}[section]
\newtheorem{remark}{Remark}[section]

\title[Cumulative Residual Mathai--Haubold Entropy]{Weighted Cumulative Residual Mathai-Haubold Entropy}

\author[]{A\lowercase{nija} C.R.\lowercase{\textsuperscript{a}} , S\lowercase{mitha} S.\lowercase{\textsuperscript{a}} \lowercase{and} S\lowercase{udheesh} K. K\lowercase{attumannil.\textsuperscript{b}}
\\
\lowercase{\textsuperscript{a}}K\lowercase{uriakose} E\lowercase{lias} C\lowercase{ollege},
  M\lowercase{annanam},
  K\lowercase{erala},
  I\lowercase{ndia,}\\
\lowercase{\textsuperscript{b}}I\lowercase{ndian} S\lowercase{tatistical} I\lowercase{nstitute},
  C\lowercase{hennai}, I\lowercase{ndia.}}
  \thanks{Corresponding author email: skkattu@isichennai.res.in}
\begin{document}
\maketitle
\doublespacing
\vspace{-0.2in}
\begin{abstract}
    In this paper, we introduce the weighted cumulative residual Mathai--Haubold entropy and establish its fundamental properties. A dynamic version is developed, and its behavior under linear transformations is studied. Bounds and explicit expressions for some lifetime distributions are derived. Characterization results based on the associated measure are obtained and two new classes of life distributions are formulated. A goodness-of-fit test for the Rayleigh distribution is proposed and its performance is evaluated through a Monte Carlo simulation study. Applications to real data sets demonstrate the practical applicability of the proposed methodology.
\end{abstract}
Keywords: Weighted cumulative residual Mathai--Haubold entropy,  Characterization, Rayleigh distribution, Monte Carlo simulation.
\section{Introduction}
Information theory is an emerging field in applied mathematics and statistics which deals with the mathematical characterization of information and uncertainty in random systems. Measures such as entropy, divergence, and mutual information are widely used to quantify uncertainty, dependence, and variation in applications including statistical inference, signal processing, and reliability analysis. The cornerstone of information theory is the entropy measure introduced by Shannon (1948). For a continuous random variable $X$ with probability density function $f(x)$, the (differential) entropy is defined as
 \begin{equation}\label{Entropy}
H(X) = - \int f(x)\,\log f(x)\,dx,
\end{equation}
where $\log$ denotes the natural logarithm. This measure provides a fundamental quantification of uncertainty and randomness associated with the random variable.
\par The Shannon entropy has been generalized in various ways to characterize different aspects of uncertainty and information in complex systems. One such generalized measure is the Mathai--Haubold entropy introduced by Mathai and Haubold (2006), which extends the classical Shannon entropy by incorporating an additional parameter. For a continuous non negative random variable $X$ with probability density function $f(x)$, the Mathai--Haubold entropy of order $\alpha$ is defined as
\begin{equation}\label{M entropy}
M_\alpha(X) = \frac{1}{\alpha - 1}\left( \int_{0}^{\infty} (f(x))^{2-\alpha}\, dx - 1  \right),~ \alpha \neq 1,\,  0<\alpha<2,
\end{equation}
and it reduces to the Shannon entropy as $\alpha \to 1$. This measure provides a generalized framework for quantifying uncertainty, where the parameter $\alpha$ controls the sensitivity of the entropy to different regions of the probability distribution, allowing more sensitive differentiation between distributions with varying degrees of randomness.
\par Ebrahimi (1996) proposed a measure that quantifies the uncertainty associated with the remaining lifetime of a system, given that it is functioning at time $t$. This measure is known as the residual entropy function. Dar and Al-Zahrani (2013) introduced the Mathai--Haubold residual entropy by extending the Mathai--Haubold entropy to the residual lifetime framework, which is defined as
\begin{align}\label{MR entropy}
M_\alpha(X_t)
&= \frac{1}{\alpha-1}\left( \int_t^\infty \left(\frac{f(x)}{\bar{F}(t)}\right)^{2-\alpha} \, dx - 1 \right),
\qquad \alpha \neq 1,\; 0<\alpha<2,
\end{align}
where $f(x)$ denotes the probability density function of $X$ and $\bar{F}(t)=1-F(t)$ is the corresponding survival function. When $t=0$, (\ref{MR entropy}) reduces to (\ref{M entropy}).
\par Rao et al. (2004) introduced the cumulative residual entropy (CRE), an uncertainty measure based on the survival function of a random variable $X$. Unlike Shannon entropy, which depends on the existence of a density function, CRE is defined in terms of the cumulative distribution function and is therefore more stable due to its integral nature.
\par Recently, Anija et al. (2025) have proposed the cumulative residual Mathai-- Haubold entropy based on (\ref{M entropy}) as
\begin{align}\label{CRM}
CRM_{\alpha}(X)
&=
\frac{1}{\alpha-1}
\left(
\int_{0}^{\infty}
(\bar{F}(x))^{\,2-\alpha}\,dx
-1
\right),~\alpha \neq 1,\; 0<\alpha<2.
\end{align}
As $\alpha \to 1$, $CRM_\alpha(X)$ becomes the cumulative residual entropy (CRE) introduced by Rao et al. (2005). They have also extended the dynamic version of (\ref{CRM}) as
\begin{align}\label{CRMt}
CRM_{\alpha}(X_t)
&=
\frac{1}{\alpha-1}
\left(
\int_{t}^{\infty}
\left(
\frac{\bar{F}(x)}{\bar{F}(t)}
\right)^{\,2-\alpha}
\,dx
-1
\right),
\qquad
\alpha \neq 1,\; 0<\alpha<2,\; t>0.
\end{align}
\par In many real-life situations, standard probability distributions may not adequately fit the data. To overcome this, weighted distributions are often used. In recent years, this approach has found applications in various fields of statistics, including family size analysis, human heredity, global population studies, renewal theory, biomedical research, statistical ecology, and reliability modeling.

Shannon entropy assigns equal weight to all events. To incorporate the relative importance of different outcomes, Belis and Guiasu (1968) proposed the concept of weighted entropy, defined as
\begin{eqnarray}\label{wentropy}
H^w(X) &=& -\int_0^\infty x f(x) \log f(x) \, dx.
\end{eqnarray}

It is a length-biased and shift-dependent measure. Building on the concept of cumulative residual entropy (CRE), Mirali \textit{et al.} (2017) introduced a new information measure, known as the weighted cumulative residual entropy (WCRE), which generalizes classical entropy by incorporating both weighting and residual lifetime information. Its dynamic version is also defined, and both are given as

\begin{equation}\label{WCRE}
\epsilon^{w}(X)
= - \int_{0}^{\infty} x \, \bar{F}(x) \, \log \bar{F}(x) \, dx,
\end{equation}
where \(\bar{F}(x) = P(X > x)\) is the survival function of \(X\).

\begin{equation}\label{DWCRE}
\epsilon^{w}(X_t)
= - \int_{t}^{\infty} x \, \frac{\bar{F}(x)}{\bar{F}(t)} \,
\log \frac{\bar{F}(x)}{\bar{F}(t)} \, dx,
\end{equation}
where \(\epsilon^{w}(X_t)\) represents the WCRE of the residual lifetime \(X-t \mid X>t\).

Khammar and Jahanshahi (2018) proposed weighted cumulative residual Tsallis entropy and studied its properties. Balakrishnan et al. (2022) introduced weighted extropy, a shift-dependent measure that assigns greater emphasis to larger values of random variables. They further extended this concept to weighted residual and past extropies, and also explored its bivariate extensions. Chakraborty and Pradhan (2023) introduced weighted cumulative Tsallis residual and past entropies in dynamic forms, studied their distributional properties, and demonstrated their applications as risk measures.
They also included empirical results and a real-data example to illustrate the practical utility.



   Motivated by the increasing relevance of weighted cumulative information measures and the limited existing work on cumulative Mathai--Haubold entropy, this paper introduces and systematically studies the weighted cumulative residual Mathai--Haubold entropy. The paper is organized as follows. Section 2 presents the definition and fundamental properties of the proposed measure. In Section 3, we introduce the dynamic weighted cumulative residual Mathai--Haubold entropy and its important properties, including linear transformations, bounds, and its forms for several well-known lifetime distributions. Section 4 provides characterization results derived from the DWCRMHE, while Section 5 defines two new classes of life distributions based on this function. In Section 6, a new goodness-of-fit test for the Rayleigh distribution is developed, followed by Monte Carlo simulation studies in Section 7. Section 8 demonstrates the practical applicability of the proposed test through analyses of two real data sets, and Section 9 concludes the paper with final remarks.
\section{Weighted Cumulative Residual Mathai-Haubold Entropy (WCRMHE)}

In lifetime studies, the survival function plays a central role. This motivates the definition of the weighted cumulative residual Mathai–Haubold entropy (WCRMHE), which is defined as follows.
\begin{definition}

Let $X$ be a continuous non-negative random variable with survival function $\bar{F}(x)$. The weighted cumulative residual Mathai--Haubold entropy (WCRMHE), denoted by $CRM^{w}_{\alpha}(X)$, is defined as

\begin{equation}\label{WCRM}
    CRM^{w}_\alpha(X)= \frac{1}{\alpha-1}\left( {\int_0^\infty} x(\bar{F}(x))^{2-\alpha} \,dx -1 \right), \alpha \neq 1, 0<\alpha<2.
\end{equation}
\end{definition}
Introducing the factor $x$ on the right-hand side of the above integral yields an information measure that is length-biased and shift-dependent, assigning greater weight to larger values of the random variable $X$.
\par We can observe that, as $\alpha \to 1$, $CRM^{w}_\alpha(X)$ reduces to $\epsilon^{w}(X)$, which is the weighted cumulative residual entropy (WCRE) of $X$, given in (\ref{WCRE}).\\

The following example illustrate the usefulness of the proposed measure by showing that the equality of the cumulative residual Mathai--Haubold entropy does not necessarily imply equality of its weighted form.\\
\begin{example}\label{ex:WCRMHE_uniform}
Consider two random variables $X$ and $Y$ with probability density functions given, respectively, by
\begin{equation*}
f(x)=
\begin{cases}
\dfrac{1}{a}, & 0<x<a,\\
0, & \text{otherwise},
\end{cases}
\qquad \text{and} \qquad
g(y)=
\begin{cases}
\dfrac{1}{a}, & b<y<a+b,\\
0, & \text{otherwise}.
\end{cases}
\end{equation*}
From~(\ref{CRM}), we obtain
\begin{equation*}
CRM_{\alpha}(X)
=
CRM_{\alpha}(Y) =
\frac{1}{\alpha-1}
\left(
\frac{a}{3-\alpha}-1
\right).
\end{equation*}
Further from (\ref{WCRM}) we get,
\begin{equation*}
CRM^{w}_{\alpha}(X)
=
\frac{1}{\alpha-1}
\left(
\frac{a^{2}}{(3-\alpha)(4-\alpha)}-1
\right),
\end{equation*}
and
\begin{equation*}
CRM^{w}_{\alpha}(Y)
=
\frac{1}{\alpha-1}
\left(
\frac{a(a+b)(4-\alpha)-a^{2}(3-\alpha)}{(3-\alpha)(4-\alpha)}-1
\right).
\end{equation*}
Although
$CRM_{\alpha}(X)=CRM_{\alpha}(Y)$, but $CRM^{w}_{\alpha}(X)\neq CRM^{w}_{\alpha}(Y).$
\end{example}
Table 1 provides some well known families of distributions and their WCRMHE.
\begin{table}[htbp]
\centering
\caption{Expression of WCRMHE for some well--known distributions.}
\label{tab:CRMHE_distributions}
\renewcommand{\arraystretch}{1.3}
\begin{tabular}{|c|c|c|c|}
\hline
 & Distribution & $\bar{F}(x)$ & $CRM_{\alpha}^{w}(X)$ \\
\hline
(i)
& Uniform distribution
& $1-\dfrac{x}{a},\; 0<x<a,\; a>0$
& $\dfrac{1}{\alpha-1}\!\left(\dfrac{a^{2}}{(3-\alpha)(4-\alpha)}-1\right)$ \\
\hline
(ii)
& Exponential distribution
& $e^{-\lambda x},\; x\geq 0,\; \lambda>0$
& $\dfrac{1}{\alpha-1}\!\left(\dfrac{1}{\lambda^{2}(2-\alpha)^{2}}-1\right)$ \\
\hline
(iii)
& Pareto distribution
& $\left(\dfrac{k}{x}\right)^{a},\; x\geq k,\; k>0,\; a>0$
& $\dfrac{1}{\alpha-1}\!\left(\dfrac{k^{2}}{a(2-\alpha)-2}-1\right)$ \\
\hline
(iv)
& Power distribution
& $1-x^{c},\; 0<x<1$
& $\dfrac{1}{\alpha-1}\left(\dfrac{1}{c}B(\tfrac{2}{c},\,3-\alpha)-1\right)$ \\
\hline
(v)
(v)
& Rayleigh distribution
& $e^{-\tfrac{x^2}{2\sigma^2}}, x\geq 0$
&$\dfrac{1}{\alpha-1}\left(\dfrac{\sigma^2}{2-\alpha}-1\right)$\\
\hline

\end{tabular}
\end{table}

The following property describes the effect of the linear transformation on the weighted cumulative residual Mathai--Haubold entropy.
\begin{property}
    Consider the random variable $Y=aX+b$ with $a>0$ and $b \geq 0$, then
    \begin{equation}\label{prop1}
      CRM_{\alpha}^{w}(Y)=\frac{a^2+ab-1}{\alpha-1}+a^2~CRM_{\alpha}^{w}(X)+ab~CRM_{\alpha}(X).
  \end{equation}
\end{property}
\begin{proof}
The proof follows from the fact that ${\bar{F}_{aX+b}}(x)={\bar{F}_{X}}(\frac{x-b}{a})$ and using (\ref{WCRM}).
\end{proof}
\begin{corollary}
It is clear from the above property that
\begin{enumerate}
\item if $a=1$, then
\[
CRM_{\alpha}^{w}(Y)
= \frac{b}{\alpha-1} + CRM_{\alpha}^{w}(X)+ b ~ CRM_{\alpha}(X),
\]

\item if $b=0$, then
\[
CRM_{\alpha}^{w}(Y)
=
\frac{a^{2}-1}{\alpha-1}
+
a^{2}~CRM_{\alpha}^{w}(X).
\]
\end{enumerate}
\end{corollary}
Let $X$ be a continuous non-negative random variable with survival function $\bar{F}(x)$ and density function $f(x)$. The weighted mean residual lifetime (WMRL) of $X$ is defined as
\begin{equation}\label{WMRL}
m_F^{*}(t)=\frac{1}{\bar{F}(t)}\int_{t}^{\infty} x\,\bar{F}(x)\,dx,
\end{equation}
and
\[
m_F^{*}(0)=\int_{0}^{\infty} x\,\bar{F}(x)\,dx.
\]
The following property establishes bounds for the weighted cumulative residual Mathai--Haubold entropy using $m_F^{*}(0)$.

\begin{property}
Let $X$ be a non--negative continuous random variable with weighted mean residual lifetime (WMRL) function $m_{F}^{*}(t)$ and weighted cumulative residual Mathai--Haubold entropy $CRM^{w}_{\alpha}(X)$ such that $CRM^{w}_{\alpha}(X)<\infty$. Then,
\[
CRM^{w}_{\alpha}(X)
> (<)
\frac{1}{\alpha-1}
\bigl(m_{F}^{*}(0)-1\bigr),
\qquad 0<\alpha < 1 (1 < \alpha<2).
\]

\end{property}
\begin{proof}
Since
\[
x(\bar{F}(x))^{2-\alpha} < (>)\, x\bar{F}(x),
\qquad \forall\,~ 0<\alpha<1 \; (1<\alpha<2).
\]
Integrating both sides with respect to $x$ gives
\[
\int_{0}^{\infty} x(\bar{F}(x))^{2-\alpha}\,dx
<
(>)
\int_{0}^{\infty} x\bar{F}(x)\,dx .
\]
Hence
\[
CRM^{w}_{\alpha}(X)
>(<)
\frac{1}{\alpha-1}\bigl(m_F^{*}(0)-1\bigr).
\]
\end{proof}
  Let $X$ be a non-negative random variable with survival function $\bar{F}(x)$. A random variable $X_{\theta}^{*}$ is said to follow the proportional hazards (PH) model with parameter $\theta>0$ if its survival function satisfies
\[
\bar{F}_{\theta}^{*}(x) = (\bar{F}(x))^{\theta}, \qquad x \ge 0,~ \theta>0.
\]

The following property gives the transformation of the weighted cumulative residual Mathai--Haubold entropy measure under the proportional hazards (PH) model.
\begin{property}
$CRM_\alpha^{w}(X_\theta^*)=\frac{\beta-1}{\alpha-1}CRM_\beta^{w}(X)$,
where $\beta=2-\theta(2-\alpha)$.
\end{property}
\begin{proof}
   From the above relation, we have

  \[
  \qquad\qquad\quad
 CRM_\alpha^{w}(X_\theta^*) =\frac{1}{\alpha-1}\left( {\int_0^\infty} x(\bar{F}(x))^{\theta(2-\alpha)} \ dx -1\right).
\]
Let $2-\beta=\theta(2-\alpha)$, which implies $\beta = 2-\theta(2-\alpha).$\\
Then
\[
CRM_\alpha^{w}(X_\theta^*)
=
\frac{1}{\alpha-1}
\left(
\int_0^\infty x(\bar{F}(x))^{2-\beta} \, dx -1
\right).
\]
Using the definition of $CRM_\beta^{w}(X)$, we get
\[
CRM_\alpha^{w}(X_\theta^*)
= \frac{\beta-1}{\alpha-1}
\,CRM_\beta^{w}(X).
\]
\end{proof}
The following example illustrates the above property for the exponential
distribution.
\begin{example}
Let
$X$ has an exponential distribution with the parameter $\lambda$. Then
\[
CRM_\alpha^{w}(X_\theta^*)
=\frac{\beta-1}{\alpha-1}~CRM_\beta^{w}(X);\beta=2-\theta(2-\alpha).
\]
\end{example}
An application of Property~2.3 to order statistics is given in the following example.
\begin{example}
Let $X$ be a non-negative continuous random variable with distribution function
$F$, and let $X_{1:n}$ denote the first order statistic based on a random sample
$X_1,X_2,\ldots,X_n$ from $F$. Then the survival function of $X_{1:n}$ is given by
\[
\bar{F}_{X_{1:n}}(x) = \bigl(\bar{F}(x)\bigr)^n.
\]
Thus, $X_{1:n}$ follows the proportional hazards model with parameter $\theta=n$.
Consequently, using the definition of the weighted cumulative residual
Mathai--Haubold entropy, we obtain
\[
CRM_\alpha^{w}(X_{1:n})
=\frac{\beta-1}{\alpha-1}\,CRM_\beta^{w}(X),
\qquad \text{where } \beta=2-n(2-\alpha).
\]
\end{example}
The following theorem provides a sufficient condition for the finiteness of $CRM_{\alpha}^{w}(X)$.
\begin{theorem}
 If there exists $p>2$ such that $E(X^p)<\infty$, then $CRM_{\alpha}^{w}(X)<\infty$.
\end{theorem}
\begin{proof}
Since $E(X^p)<\infty$ for some $p>2$, by Markov's inequality
\[
P(X>x)\le \frac{E(X^p)}{x^p}, \qquad x>0.
\]
Hence
\[
(\bar F(x))^{2-\alpha}\le \frac{(E(X^p))^{2-\alpha}}{x^{p(2-\alpha)}}.
\]
Therefore
\[
x(\bar F(x))^{2-\alpha}\le (E(X^p))^{2-\alpha}x^{1-p(2-\alpha)}.
\]
Since $p>2$ and $0<\alpha<2$
\[\int_0^\infty x(\bar F(x))^{2-\alpha}dx <\infty. \]
Thus
\[CRM_\alpha^w(X)<\infty.\]
\end{proof}


\section{Dynamic Weighted Cumulative Residual Mathai-Haubold Entropy (DWCRMHE)}
In this section, we introduce the dynamic version of weighted cumulative residual Mathai-Haubold entropy and is termed as dynamic weighted cumulative residual Mathai-Haubold entropy (DWCRMHE).
\begin{definition}
    Let $X$ be the lifetime of a component or system and assumes that it survive upto time $t$. In this situation, we consider the time--dependent or dynamic random variable $X_t=X-t|X>t$ with survival function
    \[
        \bar{F}_{t}(x) =
        \begin{cases}
          \frac{\bar{F}(x)}{\bar{F}(t)}, & \text{if $x>t$},\\
    1, & \text{$otherwise$}.
        \end{cases}
    \]
    \end{definition}
    \begin{definition}
        Let $X$ be a non--negative continuous random variable having survival function $\bar{F}(x)$. Then Dynamic weighted cumulative residual Mathai-Haubold entropy (DWCRMHE) denoted by $CRM^{w}_\alpha(F;t)$ or $CRM_{\alpha}^{w}(X_t)$, and is defined as
        \begin{align}\label{WCRMt}
     CRM^{w}_\alpha(F;t)=CRM^{w}_\alpha(X_t)= \frac{1}{\alpha-1}\left( {\int_t^\infty} x\left(\frac{\bar{F}(x)}{\bar{F}(t)}\right)^{2-\alpha} \,dx -1\right), \alpha \neq 1, 0<\alpha<2.
    \end{align}
\end{definition}
On differentiating (\ref{WCRMt}) with respect to $t$, we get
\begin{equation}\label{diff}
(\alpha-1)CRM_\alpha^{w}{'}(X_t)=(2-\alpha)h(t)\left((\alpha-1)CRM_\alpha^{w}(X_t)+1\right)-t.
\end{equation}
Next, we look into some properties of $CRM_{\alpha}^{w}(X_t)$. The following property shows that the dynamic measure reduces to the original measure at the initial time.
\begin{property}
    When $t=0$, $CRM^{w}_\alpha(X_t)=CRM^{w}_\alpha(X)$.
\end{property}
Next property gives the limiting form of the proposed measure as the parameter approaches unity.
\begin{property}
     $CRM^{w}_\alpha(X_t)$ reduces to dynamic weighted cumulative residual entropy $\epsilon^{w}(X_t)$, given in (\ref{DWCRE}) as $\alpha \to 1$.
\end{property}
The next property discusses the effect of linear transformations on $CRM^{w}_\alpha(X_t)$.
\begin{property}
    Let $Y=aX+b$ be the random variable with $a>0$ and $b \geq 0$. Then
    \begin{equation*}
        CRM^{w}_\alpha(Y;t)=\frac{a^2+ab-1}{\alpha-1}+a^2CRM_{\alpha}^{w}\left(X;\frac{t-b}{a}\right)+ab~CRM_{\alpha}\left(X;\frac{t-b}{a}\right).
    \end{equation*}
\end{property}
\begin{proof}
    We have
    \begin{equation}\label{DWCRM Y}
         CRM^{w}_\alpha(Y;t)= \frac{1}{\alpha-1}\left( {\int_t^\infty} y\left(\frac{\bar{F}_Y(y)}{\bar{F}(t)}\right)^{2-\alpha} \,dy -1\right).
    \end{equation}
    When $Y=aX+b$, (\ref{DWCRM Y}) turns into
    \begin{equation*}
        CRM^{w}_\alpha(Y;t)= \frac{1}{\alpha-1}\left(a^2\left((\alpha-1)CRM_{\alpha}^{w}\left(X;\frac{t-b}{a}\right)+1\right)+ab\left((\alpha-1)CRM_{\alpha}\left(X;\frac{t-b}{a}\right)+1\right)-1\right),
    \end{equation*}
    which implies
    \begin{equation*}
         CRM^{w}_\alpha(Y;t)=\frac{a^2+ab-1}{\alpha-1}+a^2CRM_{\alpha}^{w}\left(X;\frac{t-b}{a}\right)+abCRM_{\alpha}\left(X;\frac{t-b}{a}\right).
    \end{equation*}
   This completes the proof.
\end{proof}
\begin{corollary}
  It is clear from the above property that\\
    $(i)$ if $a=1$, then $CRM^{w}_\alpha(X+b;t)=\frac{b}{\alpha-1}+CRM_{\alpha}^{w}\left(X;{t-b}\right)+b~CRM_{\alpha}\left(X;{t-b}\right)$\\
    $(ii)$ if $b=0$, then $CRM_{\alpha}^{w}(aX;t)=\frac{a^2-1}{\alpha-1}+a^2~CRM_{\alpha}^{w}\left(X;\frac{t}{a}\right)$.
\end{corollary}
The following theorem establishes the uniqueness property of the proposed dynamic measure.
\begin{theorem}
 Let $X$ be a non-negative random variable with density function $f(x)$, survival function $\bar{F}(x)$ and hazard rate $h(x)$ and consider $CRM_{\alpha}^{w}(X_t)$ to increase in $t$. Then $CRM_{\alpha}^{w}(X_t)$ uniquely determines its distribution function.
\end{theorem}
\begin{proof}
    Let us assume that $F(x)$ and $G(x)$ are the distribution functions such that
    \begin{equation}\label{char1}
       CRM_{\alpha}^{w}(F;t) = CRM_{\alpha}^{w}(G;t), \forall~ t\geq 0.
    \end{equation}
    This implies that
    \begin{equation}\label{char11}
       \int_{t}^{\infty}{x{\left(\frac{\bar{F}(x)}{\bar{F}(t)}\right)}^{2-\alpha}dx}=\int_{t}^{\infty}{x{\left(\frac{\bar{G}(x)}{\bar{G}(t)}\right)}^{2-\alpha}dx}.
    \end{equation}
    Differentiating (\ref{char1}) with respect to $t$, we get
    \begin{equation}\label{char12}
        h_1(t)((\alpha-1)CRM_{\alpha}^{w}(F;t)+1)=h_2(t)((\alpha-1)CRM_{\alpha}^{w}(G;t)+1),
    \end{equation}
    where $h_1(t)$ and $h_2(t)$ are the hazard rates related to $f(x)$ and $g(x)$ respectively.\\
    Using (\ref{char1}) in (\ref{char12}), we get
    \begin{equation*}
        h_1(t)=h_2(t)
    \end{equation*}
    which implies that
    \begin{equation*}
        \bar{F}(t)=\bar{G}(t).
    \end{equation*}
    Hence $CRM_{\alpha}^{w}(X_t)$ uniquely determines its distribution function.
\end{proof}
Table \ref{tab:WCRMHE_distributions} provides some well known families of distributions and their DWCRMHE.
\begin{table}[htbp]
\centering
\caption{Expression of DWCRMHE for some well--known distributions.}
\label{tab:WCRMHE_distributions}
\renewcommand{\arraystretch}{1.3}
\begin{tabular}{|c|c|c|c|}
\hline
 & Distribution & $\bar{F}(x)$ & $CRM_{\alpha}^{w}(X_t)$ \\
\hline
(i)
& Uniform distribution
& $1-\dfrac{x}{a},\; 0<x<a,\; a>0$
& $\dfrac{1}{\alpha-1}\!\left[\dfrac{a(a-t)}{(3-\alpha)}-\dfrac{(a-t)^2}{4-\alpha}-1\right]$ \\
\hline
(ii)
& Exponential distribution
& $e^{-\lambda x},\; x\geq 0,\; \lambda>0$
& $\dfrac{1}{\alpha-1}\!\left[\dfrac{t}{\lambda(2-\alpha)}+\dfrac{1}{\lambda^{2}(2-\alpha)^{2}}-1\right]$ \\
\hline
(iii)
& Pareto distribution
& $\left(\dfrac{k}{x}\right)^{a},\; x\geq k,\; k>0,\; a>0$
& $\dfrac{1}{\alpha-1}\!\left[\dfrac{\left(\frac{t}{max(t,k)}\right)^{a(2-\alpha)}(max(t,k))^{2}}{a(2-\alpha)-2}-1\right]$ \\
\hline
(iv)
& Power distribution
& $1-x^{c},\; 0<x<1$
& $\dfrac{1}{\alpha-1}\left[\dfrac{1}{c(1-t^c)^{2-\alpha}}B_{1-t^{c}}(3-\alpha, \tfrac{2}{c})-1\right]$ \\
\hline
(v)
& Rayleigh distribution
& $e^{-\tfrac{x^2}{2\sigma^2}}, x\geq 0$
&$\dfrac{1}{\alpha-1}\left[\dfrac{\sigma^2}{2-\alpha}-1\right]$\\
\hline

\end{tabular}
\end{table}
The following section presents the characterization results associated with $CRM^{w}_\alpha(X_t)$.
\section{Characterization Results}
Through studying the behaviour and properties of the dynamic weighted cumulative residual Mathai--Haubold entropy, we introduce the characterization results which uniquely determine the associated probability distribution.
The following theorem characterizes the Rayleigh distribution using $CRM_{\alpha}^{w}(X_t)$.
\begin{theorem}\label{cons}
    For a continuous non--negative random variable $X$ with distribution function $F(x)$, $CRM_{\alpha}^{w}(X_t)$ is independent of $t$ if and only if $X$ has the Rayleigh distribution.
\end{theorem}
\begin{proof}
   Assume that
   \[CRM_{\alpha}^{w}(X_t)=k, \text {a constant}.\]
Therefore
\[ CRM_{\alpha}^{w}{'}(X_t)=0. \]
   Using (\ref{diff}), we get
   \begin{equation*}
       (2-\alpha) ~h(t)\left((\alpha-1)~CRM_\alpha^{w}(X_t)+1\right)-t=0.
   \end{equation*}
   This implies
   \begin{equation*}
       h(t)=\frac{t}{\sigma^2}; \sigma^2=(2-\alpha)((\alpha-1)k+1),
   \end{equation*}
which is the hazard rate of the Rayleigh distribution.\\ Conversely, let $X$ follows the Rayleigh distribution with the survival function
\begin{equation*}
\bar{F}(x)=exp\left({-\frac{x^2}{2\sigma^2}}\right), x\geq 0, \sigma>0.\\
\end{equation*}
 Then
   \begin{equation*}
       CRM_{\alpha}^{w}(X_t)=\frac{1}{\alpha-1}\left(\frac{\sigma^2}{2-\alpha}-1\right),\text{a constant.}
   \end{equation*}
   Hence, the theorem.
\end{proof}
Next theorem characterizes the Rayleigh distribution in terms of $CRM_{\alpha}^{w}(X_t)$ and the weighted mean residual life function $m^{*}_{F}(t)$.
\begin{theorem}
   Consider a continuous non--negative random variable $X$. Then
\[
CRM_{\alpha}^{w}(X_t)=\frac{1}{\alpha-1}\left(\frac{m^{*}_{F}(t)}{2-\alpha}-1\right),
\]
where, $m^{*}_{F}(t)$ is the WMRL given in (\ref{WMRL}),
if and only if $X$ follows the Rayleigh distribution with survival function
\[
\bar{F}(t)=\exp\left(-\frac{t^2}{2\sigma^2}\right), \quad \sigma>0.
\]
\end{theorem}
\begin{proof}
 Assume \begin{equation*}
     CRM_{\alpha}^{w}(X_t)=\frac{1}{\alpha-1}\left(\frac{m^{*}_{F}(t)}{2-\alpha}-1\right).
  \end{equation*}
  Using (\ref{diff}), we get
  \begin{equation}\label{thm1}
    h(t)~m^{*}_{F}(t)-t=\frac{1}{2-\alpha}m_{F}^{*}{'}(t).
  \end{equation}
  Differentiating $m^{*}_{F}(t)$ with respect to $t$, we get
  \begin{equation}\label{thm2}
      m_{F}^{*}{'}(t)=m^{*}_{F}(t)~h(t)-t.
  \end{equation}
  From (\ref{thm1}) and (\ref{thm2}), we get
  \begin{equation*}
       m^{*}{'}_{F}(t)=0,
  \end{equation*}
  which implies
  \begin{equation*}
      m^{*}_{F}(t)=\sigma^2, \text{a constant}.
  \end{equation*}
  Also, from (\ref{thm1}), we get
  \begin{equation*}
      h(t)~m^{*}_{F}(t)=t.
  \end{equation*}
  So
  \begin{equation*}
      h(t)=\frac{t}{\sigma^2}.
  \end{equation*}
  Therefore
  \begin{equation*}
      \bar{F}(t)=exp\left(-\frac{t^2}{2\sigma^2}\right),
  \end{equation*}
  which is the survival function of the Rayleigh distribution.
  Conversely, assume that $X$ has the Rayleigh distribution. Then
  \begin{equation*}
      CRM_{\alpha}^{w}(X_t)=\frac{1}{\alpha-1}\left(\frac{\sigma^2}{2-\alpha}-1\right)=\frac{1}{\alpha-1}\left(\frac{m^{*}_{F}(t)}{2-\alpha}-1\right).
  \end{equation*}
  Hence, the proof.
\end{proof}
\section{New classes of lifetime distributions}
In this section, we introduce two new classes of lifetime distributions based on the dynamic weighted cumulative residual Mathai-- Haubold entropy.
\begin{definition}\label{def}
    The random variable $X$ is said to have an increasing (decreasing) WDCRMHE, denoted by IWDCRMHE (DWDCRMHE), if $CRM_{\alpha}^{w}(X_t)$ is increasing (decreasing) as time $t$ passes for $t\geq 0$.

\end{definition}
\begin{definition}
   The random variable $X$ is said to have an increasing (decreasing) failure rate IFR (DFR), if $h(t)$ is increasing (decreasing) as time $t$ passes for $t\geq 0$.
\end{definition}
\begin{theorem}
    For a random variable $X$ with distribution function $F(x)$ is said to have an increasing (decreasing) WDCRMHE if and only if $\forall~ t>0$.
    \begin{equation}\label{new1}
        CRM_\alpha^{w}(X_t)\geq(\leq)\frac{1}{\alpha-1}\left(\frac{t}{(2-\alpha)h(t)}-1\right), for~ 1<\alpha<2
    \end{equation}
  and
  \begin{equation}\label{new2}
        CRM_\alpha^{w}(X_t)\leq(\geq)\frac{1}{\alpha-1}\left(\frac{t}{(2-\alpha)h(t)}-1\right), for ~ 0<\alpha<1.
    \end{equation}
\end{theorem}
\begin{proof}
    The proof directly follows from (\ref{diff}) and Definition \ref{def}.
\end{proof}
\begin{remark}
  Equations (\ref{new1}) and (\ref{new2}) can also be expressed as
\end{remark}
\begin{equation*}
    h(t)\geq(\leq)\frac{t}{(2-\alpha)((\alpha-1)CRM_\alpha^{w}(X_t)+1)}, for~ 1<\alpha<2
  \end{equation*}
  and
  \begin{equation*}
      h(t)\leq(\geq)\frac{t}{(2-\alpha)((\alpha-1)CRM_\alpha^{w}(X_t)+1)}, for~ 0<\alpha<1.
  \end{equation*}
  \begin{theorem}
      Let $X$ and $Y$ be two non--negative continuous random variables with survival functions $\bar{F}(t)$ and $\bar{G}(t)$ and hazard rate functions $h_{1}(t)$ and $h_{2}(t)$, respectively. Suppose $X\geq^{hr} Y$, that is, $h_{1}(t) \leq h_{2}(t)$ for all $t\geq 0$.
      Then $CRM_\alpha^{w}(F;t) \geq(\leq)~ CRM_\alpha^{w}(G;t)$ $\forall$ $1<\alpha<2$~($0<\alpha<1$).
  \end{theorem}
  \begin{proof}
       Let $h_1(t)\leq h_2(t)$.
        Then $\bar{F}_{X_t}(t)$$\geq$ $\bar{G}_{X_t}(t)$.\\
       So,
        \[
        \frac{\bar{F}(x)}{\bar{F}(t)} \geq \frac{\bar{G}(x)}{\bar{G}(t)}.
        \]
        \[
        x{\left(\frac{\bar{F}(x)}{\bar{F}(t)}\right)}^{(2-\alpha)} \geq x{\left( \frac{\bar{G}(x)}{\bar{G}(t)}\right)}^{(2-\alpha)}        \]
        Integrating both sides, we get
        \[
       \int_{t}^{\infty} x{\left(\frac{\bar{F}(x)}{\bar{F}(t)}\right)}^{(2-\alpha)} dx\geq \int_{t}^{\infty}x{\left( \frac{\bar{G}(x)}{\bar{G}(t)}\right)}^{(2-\alpha)}dx.
        \]
        From (\ref{WCRMt}), we get
    \[
        (\alpha-1)CRM_\alpha^{w}(F;t)+1\geq(\alpha-1)CRM_\alpha^{w}(G;t)+1.
     \]
     Therefore

        \[
        CRM_\alpha^{w}(F;t) \geq CRM_\alpha^{w}(G;t), for~ 1<\alpha<2\]
      and
      \[CRM_\alpha^{w}(F;t) \leq CRM_\alpha^{w}(G;t), for~ 0<\alpha<1.\]
  \end{proof}
  \begin{theorem}
      Let $X$ be a random variable that possesses both IWDCRMHE and DWDCRMHE. Then $X$ has a Rayleigh distribution.
  \end{theorem}
  \begin{proof}
      Let $X$ possess IWDCRMHE (DWDCRMHE). Then
      \begin{equation}\label{ray}
      CRM_{\alpha}^{w}{'}(X_t)\geq(\leq)0.
      \end{equation}
     From (\ref{ray}), it is clear that
     \begin{equation*}
     CRM_{\alpha}^{w}(X_t)=k, a~ constant.
     \end{equation*}
     From Theorem \ref{cons}, we can see that the constancy of $CRM_\alpha^{w}(X_t)$ implies that $X$ has the Rayleigh distribution.Hence, the proof.

  \end{proof}
  \section{Test for Rayleigh distribution}
We develop a test for the Rayleigh distribution using a characterisation result stated in Theorem 4.1.   Let $X_1, X_2, \ldots, X_n$ be a random sample of size $n$ drawn from an unknown distribution function $F$. The objective is to test whether the underlying distribution $F$ belongs to the Rayleigh family. Accordingly, we formulate the hypotheses as follows:
\begin{equation*}
H_0 : X~ \text{follows a Rayleigh distribution}
\end{equation*}
against
\begin{equation*}
H_1 : X~ \text{does not follow a Rayleigh distribution.}
\end{equation*}
   For testing the above hypothesis, we define a departure measure, which discriminates between the null and the alternative hypotheses.\\
 For the Rayleigh distribution, $CRM_\alpha^{w}(X_t)$ is independent of $t$. Therefore, under the null hypothesis
  \[
  CRM_\alpha^{w}{'}(X_t)=0, \text{for all $t$} \geq0,
  \]
    which implies
  \begin{equation*}
      t{(\bar{F}(t))^{(3-\alpha)}}-(2-\alpha)f(t)\int_{t}^{\infty}x{(\bar{F}(x))^{(2-\alpha)}}\, dx=0.
  \end{equation*}
  In view of the above identity, we consider a measure of departure $\Delta(F)$, which is given by
  \begin{equation}\label{departure}
      \Delta(F)= \int_{0}^{\infty}\left(t{(\bar{F}(t))^{(3-\alpha)}}-(2-\alpha)f(t)\int_{t}^{\infty}x{(\bar{F}(x))^{(2-\alpha)}}\, dx\right)\, dt.
  \end{equation}
  Clearly, under $H_0$, $\Delta(F)$ is zero and under $H_1$, $\Delta(F)$ is positive. Therefore, $\Delta(F)$ can be considered as the departure measure from the null hypothesis $H_0$ to the alternative hypothesis $H_1$.    Using integration by parts on (\ref{departure}), we get
   \begin{align}\label{departure2}
       \Delta(F)=\frac{(3-\alpha)^2}{2}\int_{0}^{\infty}t^2(\bar{F}(t))^{(2-\alpha)}f(t)dt-\frac{(2-\alpha)^2}{2}\int_{0}^{\infty}t^2(\bar{F}(t))^{(1-\alpha)}f(t)dt.
   \end{align}
   Next, we obtain test statistics based on $\Delta(F)$. For this purpose, consider the order statistics $X_{(1)}$, $X_{(2)}$,...,$X_{(n)}$ based on a random sample of size $n$. The empirical survival function is given by
\[
\bar{F}_n\!\left(X_{(i)}\right)=\frac{n-i}{n}, \qquad i=1,2,\ldots,n.
\]
Hence, the test statistic can be defined as a plug-in estimator of $\Delta(F)$, which is given by
   \begin{align}\label{departure3}
     \widehat{\Delta}=\frac{1}{2n}\sum_{i=1}^{n}X_{(i)}^2\left((3-\alpha)^2\left(\frac{n-i}{n}\right)^{(2-\alpha)}-(2-\alpha)^2\left(\frac{n-i}{n}\right)^{(1-\alpha)}\right).
   \end{align}We reject $H_0$ against the alternative $H_1$ for large values of  $\widehat{\Delta}.$
   Next, we look into the asymptotic distribution of $ \widehat{\Delta}$.
   \begin{theorem}
 As $n \to \infty$, $\sqrt{n}(\widehat\Delta-\Delta(F))$ converges in distribution to a normal random variable with mean zero and variance $\sigma^2$, where $\sigma^{2}$ is given by \begin{equation*}
   \sigma^{2}=Var \Big[ \frac{1}{2}X^2\Big ((3-\alpha)^2\bar{F}^{2-\alpha}(x)-(2-\alpha)^2\bar{F}^{1-\alpha}(x)\Big)\Big]
\end{equation*}
and
\begin{equation*}
\Delta(F)=E(h(X)).
\end{equation*}
\end{theorem}
\begin{proof}
Consider
\[
h_n(x)=\frac{1}{2}x^2\Big((3-\alpha)^2\bar{F}_n^{\,2-\alpha}(x)
-(2-\alpha)^2\bar{F}_n^{\,1-\alpha}(x)\Big).
\]
Since
\[
\bar{F}_n\!\left(X_{(i)}\right)=\frac{n-i}{n},
\]
the test statistic can be written as
\[
\widehat{\Delta}
=\frac{1}{n}\sum_{i=1}^{n} h_n\!\left(X_{(i)}\right).
\]
Using the Glivenko--Cantelli theorem,
\[
\sup_x \bigl|\bar{F}_n(x)-\bar{F}(x)\bigr|
\xrightarrow{\text{a.s.}} 0 .
\]
Therefore
\[
h_n(x) \xrightarrow{a.s.} h(x).
\]
Hence
\[
\widehat{\Delta}
=\frac{1}{n}\sum_{i=1}^{n} h_n(X_i)
=\frac{1}{n}\sum_{i=1}^{n} h(X_i) + o_p(1).
\]

Thus, the estimator behaves asymptotically like the sample mean of
\(h(X_i)\).
       Also, since $X_1,X_2,\ldots,X_n$ are i.i.d., it follows that
$h(X_1),h(X_2),\ldots,h(X_n)$ are also i.i.d. random variables with
\[
{E}\bigl[h(X)\bigr]=\Delta(F)
\quad \text{and} \quad
{Var}\bigl(h(X)\bigr)=\sigma^{2}.
\]
By the Central Limit Theorem,
\begin{align*}
\sqrt{n}\left(
\frac{1}{n}\sum_{i=1}^{n} h(X_i)-\Delta(F)
\right)
\xrightarrow{d} N(0,\sigma^{2}).
\end{align*}which prove the required results.

   \end{proof}

\begin{corollary}\label{test}
Under $H_0$, as  $n \to \infty$,  $\sqrt{n}\,\widehat{\Delta}$ converges in distribution to a normal random variable with mean zero and variance $\sigma_0^{2}$, where $\sigma_0^{2}$ is the value of $\sigma^{2}$ evaluated under $H_0$.
\end{corollary}

Using Corollary $\ref{test}$, we can obtain the asymptotic critical region.
We reject the null hypothesis $H_0$ against the alternative hypothesis $H_1$ at a significance level $\alpha$, if
\begin{equation*}
    \left|\frac{\sqrt{n}\widehat\Delta}{\widehat\sigma_{0}}\right| > z_{\alpha/2},
\end{equation*}
where $z_{\alpha/2}$ be the upper $\alpha/2$--percentile point of the standard normal distribution.

 \section{Simulation}

   To assess the finite sample performance of the proposed test,  we put forward an extensive Monte Carlo simulation study with the empirical power of the test. We used the R software to perform the simulation 10000 times using different sample sizes.\\
 We compare empirical power of the proposed test with the classical goodness of fit tests which include Anderson Darling (AD) test, Cramer von Mises (CvM) test and Kolmogrov--Smirnov (KS) test, and other goodness of fit tests which are particularly introduced for Rayleigh distribution. For different values of the parameter $\alpha$, we perform the proposed test. When $\alpha=0.1$, we observe that the proposed test performs better. So, we choose the value of $\alpha$ as 0.1 to compare with other competent tests. To obtain the empirical power of the proposed test and other competent tests, we generate lifetimes of different sample sizes, $n= 20, 30, 40, 50, 60$.

\par Since it is difficult to find a consistent estimator for the null variance ${\sigma_0}^2$, for the proposed test, we find the critical region using Monte Carlo simulation. We calculate the critical point $c_1$ as $P(\widehat\Delta>c_1)=\alpha'$, where $\alpha'$ denotes the significance level. The empirical power of the proposed test can be calculated using the algorithm given below.\\
\begin{enumerate}[noitemsep, topsep=0pt]
    \item From the desired alternative, first generate the lifetime data and obtain $\widehat\Delta$.
    \item From the standard Rayleigh distribution, generate lifetime data and obtain the value of $\widehat\Delta$.
    \item Repeat the previous step 10000 times and calculate the critical point.
    \item Repeat step (1) 10000 times and as the proportion ($p$) of significant test statistics, obtain the empirical power.
\end{enumerate}
We generate the lifetime random variables from different alternatives which include exponential, linear failure rate (LFR), Makeham, halfnormal and Pareto distributions to obtain the empirical power. The corresponding distribution functions are as follows.\\
\begin{itemize}[noitemsep, topsep=0pt]
\item Exponential distribution
\begin{equation*}
   F(x) = 1 - e^{-\lambda x}, \quad x > 0,\; \lambda > 0.
\end{equation*}
\item LFR distribution
\begin{equation*}
F(x)=1-\exp(-ax-\frac{b x^2}{2}), x > 0, a\geq 0, b\geq 0.
\end{equation*}
\item Makeham distribution
\begin{equation*}
F(x)=1-\exp(-ax-\frac{b}{c}(e^{cx}-1)), x\geq 0, a> 0, b> 0, c>0.
\end{equation*}
\item Halfnormal distribution
\begin{equation*}
    F(x)=erf\left(\frac{x}{\sigma \sqrt{2}}\right), \sigma >0, \text{where $erf$ be the error function}.
    \end{equation*}
    \item
    Pareto distribution
    \begin{equation*}
        F(x)=1-\left(\frac{\lambda}{x}\right)^a, x>0, a, \lambda>0.
\end{equation*}
\end{itemize}

    \begin{table}[H]

\centering
\caption{Empirical power comparison for various distributions when the value of the parameter $\alpha=0.1$, $\alpha=0.5$ and $\alpha=0.9$}
\label{tab:table3}
\scalebox{0.9}{
\begin{tabular}{llcccccc}
\toprule
 & &\multicolumn{2}{c}{$\alpha=0.1$} & \multicolumn{2}{c}{$\alpha=0.5$} & \multicolumn{2}{c}{$\alpha=0.9$}  \\
\cmidrule(lr){3-4} \cmidrule(lr){5-6} \cmidrule(lr){7-8}
Distribution & $n$
& $\alpha'=0.01$ & $\alpha'=0.05$
& $\alpha'=0.01$ & $\alpha'=0.05$
& $\alpha'=0.01$ & $\alpha'=0.05$
\\
\midrule

Rayleigh(1)

&20  & 0.0101 &0.0512& 0.0136 &0.045 &0.0092 &0.0525\\
&30 &  0.0091& 0.0498& 0.0078& 0.0486 &0.0095& 0.0478\\
&40  & 0.0109& 0.0555& 0.0124 &0.0507& 0.0110& 0.0551\\
&50&   0.0096 &0.0404 &0.0081 &0.0528& 0.0076& 0.0434 \\
&60 &  0.0112& 0.0532 &0.0117& 0.0471& 0.0092& 0.0485\\

\midrule
Exponential(2)
&20 &  0.9891 &0.9972 &0.9678& 0.9906 &0.9344& 0.9696 \\
&30 &  0.9986& 0.9999 &0.9957 &0.9993 &0.9735 &0.9927\\
&40  & 1.0000& 1.0000& 0.9995& 0.9999& 0.9925& 0.9978 \\
&50&   1.0000& 1.0000& 0.9999& 1.0000& 0.9966& 0.9992 \\
&60 &  1.0000& 1.0000& 1.0000& 1.0000& 0.9993& 0.9999 \\

\midrule
LFR(1,1)
&20&   0.9179& 0.9779 &0.8763 &0.9618 &0.7849& 0.9256 \\
&30&   0.9828& 0.9968& 0.9660 &0.9934& 0.9257& 0.9818 \\
&40&   0.9971& 0.9997& 0.9920& 0.9990& 0.9712 &0.9957 \\
&50&   0.9999& 0.9999 &0.9975& 0.9998 &0.9934 &0.9986 \\
&60  & 1.0000& 1.0000& 0.9999& 1.0000& 0.9979& 0.9997 \\

\midrule
Makeham(1,0.5,0.5)

&20&   0.9452& 0.9845& 0.8968& 0.9652 &0.8168 &0.9176 \\
&30&   0.9878 &0.9975 &0.9673& 0.9921& 0.9216 &0.9724 \\
&40 &  0.9985& 0.9998 &0.9929& 0.9981& 0.9623& 0.9900\\
&50 &  0.9997& 0.9999& 0.9980& 1.0000& 0.9844 &0.9975 \\
&60 &  1.0000& 1.0000& 0.9991& 1.0000& 0.9936& 0.9984\\

\midrule
Halfnormal(1)
&20 &  0.5760& 0.7888& 0.4722& 0.7138 &0.3939 &0.6414\\
&30  & 0.7446& 0.8925& 0.6452& 0.8467 &0.5104& 0.744 \\
&40&   0.8540& 0.9445& 0.7738& 0.9151& 0.6438 &0.8397 \\
&50&   0.9174 &0.9763& 0.8486& 0.9468 &0.7155& 0.8805 \\
&60 &  0.9600& 0.9903& 0.9077& 0.9715& 0.7999& 0.9304 \\
\midrule
Pareto(1,5)

&20  & 0.9353 &0.9846 &0.9204 &0.9770 &0.8950& 0.9577 \\
&30 &  0.9949& 0.9991 &0.9874 &0.9968 &0.9625& 0.9890\\
&40 &  0.9996 &1.0000& 0.9964 &0.9991& 0.9876& 0.9958 \\
&50 &  0.9999 &1.0000& 0.9995& 0.9998& 0.9940& 0.9983 \\
&60&   1.0000& 1.0000& 0.9995& 1.0000& 0.9973& 0.9991\\

\bottomrule
\end{tabular}}
\end{table}

Here, we perform the proposed test for different values of the parameter $\alpha$. From Table \ref{tab:table3}, we can see that the empirical power decreases as the value of $\alpha$ increases. Hence, for $\alpha=0.1$, the proposed test gives a high power.
\par The performance of the proposed test (denoted by ASS) is compared with the goodness of fit test for the Rayleigh distribution introduced by Jahanshahi et al. (2016) (JH), Baratpour and Khodadadi (2013) (BK), and Finkelstein and Schafer (1971) (FS), as well as with the classical tests Kolmogrov--Smirnov (KS), Cramer von Mises (CvM) and Anderson Darling (AD) test.

\begin{table}[!htbp]
\centering
\caption{Empirical power comparison for different alternatives ($\alpha' = 0.01$)}
\label{0.01}
\begin{tabular}{c c ccccccc}
\hline
Distribution & $n$ & ASS  &   KS   & CvM  &   AD     &   JH   &  BK   &  FS\\
\hline

\multirow{5}{*}{Rayleigh(1)}

&20 &0.0101& 0.0083& 0.0104& 0.0103 &0.0129 &0.0154& 0.0106\\
&30& 0.0117& 0.0106& 0.0108& 0.0095&  0.0076& 0.0118 &0.0106\\
&40& 0.0082 &0.0077& 0.0074& 0.0073&  0.0100 &0.0101& 0.0067\\
&50& 0.0109& 0.0117& 0.0131& 0.0113 & 0.0103& 0.0110& 0.0105\\
&60 &0.0111& 0.0115& 0.0103& 0.0088 & 0.0086& 0.0116 &0.0076\\

\midrule
\multirow{5}{*}{Exponential(2)}

&20& 0.9893& 0.4857& 0.4315& 0.9059 & 0.8799& 0.7711 &0.8075\\
&30& 0.9989& 0.6428& 0.5791& 0.9823&  0.9369& 0.9158 &0.9422\\
&40 &1.0000 &0.7672 &0.7110 &0.9968 & 0.9577& 0.9639 &0.9859\\
&50& 1.0000& 0.8459 &0.8591& 0.9991&  0.9690& 0.9894 &0.9960\\
&60& 1.0000& 0.9128& 0.8965& 1.0000&  0.9794& 0.9089& 0.9994\\
\midrule
\multirow{5}{*}{LFR(1,1)}
&20& 0.9053& 0.4377& 0.3805& 0.5879 & 0.5812 &0.3617 &0.3856\\
&30& 0.9819 &0.6352& 0.5570 &0.7644&  0.6573& 0.4816& 0.5810\\
&40& 0.9962& 0.7720 &0.7385& 0.8759&  0.6991& 0.6140& 0.7113\\
&50& 0.9998& 0.8885& 0.8599& 0.9283 & 0.7190& 0.7030& 0.8213\\
&60& 1.0000& 0.9533 &0.9349& 0.9692 & 0.7443& 0.3288 &0.8997\\

\midrule

\multirow{5}{*}{Makeham(1,0.5,0.5)}
&20 &0.9390 &0.3528& 0.3056& 0.8286&  0.8135& 0.6299 &0.6699\\
&30& 0.9884& 0.5408& 0.4796& 0.9504&  0.8726 &0.8149 &0.8629\\
&40& 0.9974& 0.6551& 0.6171 &0.9892&  0.9024 &0.9134 &0.9575\\
&50 &1.0000 &0.7658& 0.7484& 0.9950&  0.9406& 0.9466 &0.9792\\
&60& 1.0000& 0.8930& 0.8892& 0.9993 & 0.9537& 0.7697& 0.9959\\
\midrule
\multirow{5}{*}{Halfnormal(1)}
&20 &0.5628 &0.1686& 0.1432 &0.5132 & 0.5198 &0.2765& 0.3211\\
&30& 0.7332& 0.2925& 0.2503& 0.7222 & 0.6301 &0.4215 &0.4998\\
&40 &0.8441& 0.3945& 0.3442& 0.8313&  0.6518 &0.4954 &0.6617\\
&50 &0.9222 &0.4734 &0.4415& 0.9059&  0.6918 &0.6044& 0.7651\\
&60 &0.9543& 0.6041& 0.5968 &0.9522 & 0.7243 &0.2574& 0.8504\\
\midrule
\multirow{4}{*}{Pareto(1,5)}
&20& 0.9312& 0.9519 &0.8227 &0.9981 & 0.8820 &0.5591 &0.9925\\
&30& 0.9940& 0.9981& 0.9802 &1.0000&  0.7977& 0.8674 &0.9998\\
&40& 0.9990& 1.0000& 0.9990& 1.0000 & 0.7289 &0.9710 &1.0000\\
&50& 1.0000& 1.0000& 1.0000 &1.0000&  0.6886& 0.9928 &1.0000\\
&60& 1.0000& 1.0000& 1.0000& 1.0000 & 0.6414& 0.5101 &1.0000\\

\hline
\end{tabular}
\end{table}

\begin{table}[!htbp]
\centering
\caption{Empirical power comparison for different alternatives ($\alpha' = 0.05$)}
\label{0.05}
\begin{tabular}{c c ccccccc}
\hline
Distribution & $n$ & ASS  &   KS   & CvM  &   AD     &   JH   &  BK   &  FS\\
\hline

\multirow{5}{*}{Rayleigh(1)}

&20& 0.0491& 0.0490& 0.0536& 0.0538&  0.0486& 0.0567& 0.0523\\
&30& 0.0518& 0.0476& 0.0488& 0.0487& 0.0430 &0.0541& 0.0490\\
&40& 0.0497& 0.0436& 0.0454& 0.0438&  0.0475& 0.0441& 0.0446\\
&50& 0.0522& 0.0516& 0.0522& 0.0492&  0.0528& 0.0541& 0.0480\\
&60& 0.0528& 0.0501& 0.0496& 0.0463&  0.0457& 0.0556& 0.0471\\

\midrule
\multirow{5}{*}{Exponential(2)}
&20& 0.9972 &0.7089& 0.7042& 0.9586&  0.9526& 0.8832 &0.8996\\
&30& 0.9998 &0.8460& 0.8444& 0.9927&  0.9771& 0.9713& 0.9801\\
&40& 1.0000 &0.9186& 0.9344 &0.9991 & 0.9870& 0.9918& 0.9961\\
&50& 1.0000& 0.9705 &0.9800& 0.9999&  0.9903& 0.9975& 0.9990\\
&60& 1.0000& 0.9855& 0.9912& 1.0000 & 0.9953& 0.9636 &1.0000\\

\midrule
\multirow{5}{*}{LFR(1,1)}
&20& 0.9796& 0.6916& 0.6803 &0.7456&  0.7684& 0.5539& 0.5753\\
&30& 0.9971& 0.8519 &0.8455 &0.8758& 0.8153 &0.6897 &0.7517\\
&40& 0.9995& 0.9331& 0.9279& 0.9458 & 0.8374& 0.7916& 0.8623\\
&50& 0.9999& 0.9736& 0.9754& 0.9759 & 0.8607& 0.8634 &0.9236\\
&60& 1.0000& 0.9919& 0.9923 &0.9913 & 0.8878& 0.5191& 0.9647\\
\midrule

\multirow{5}{*}{Makeham(1,0.5,0.5)}
&20& 0.9832& 0.6210& 0.6069& 0.9160 & 0.9162 &0.8008 &0.8245\\
&30& 0.9978& 0.7857& 0.8031& 0.9819&  0.9563& 0.9248 &0.9463\\
&40& 0.9996& 0.9006& 0.9039& 0.9954&  0.9644& 0.9695& 0.9838\\
&50& 0.9998& 0.9538& 0.9637& 0.9994 & 0.9789& 0.9880& 0.9957\\
&60& 1.0000 &0.9820& 0.9862& 0.9998&  0.9841& 0.8809& 0.9985\\
\midrule
\multirow{5}{*}{Halfnormal(1)}
&20& 0.7734& 0.3843& 0.3743 &0.7025 & 0.7263 &0.4881 &0.5114\\
&30& 0.8971& 0.5704 &0.5659& 0.8516 & 0.7946 &0.6316& 0.7092\\
&40& 0.9487& 0.6983& 0.7206& 0.9200 & 0.8129& 0.7286& 0.8148\\
&50& 0.9773& 0.7856& 0.8031& 0.9629 & 0.8359 &0.8168& 0.8892\\
&60 &0.9880& 0.8608& 0.8799 &0.9814 & 0.8675& 0.4489 &0.9363\\
\midrule
\multirow{4}{*}{Pareto(1,5)}
&20& 0.9825& 0.9945& 0.9773 & 1.0000&  0.9448& 0.9300& 0.9999\\
&30 &0.9990& 1.0000& 0.9997&  1.0000&  0.8405& 0.9924& 1.0000\\
&40& 0.9999& 1.0000 &1.0000  &1.0000&  0.7770 &0.9995 &1.0000\\
&50 &1.0000& 1.0000& 1.0000&  1.0000&  0.7308& 1.0000& 1.0000\\
&60& 1.0000 &1.0000 &1.0000 & 1.0000&  0.6887& 0.7731& 1.0000\\

\hline
\end{tabular}
\end{table}
Table \ref{0.01} and Table \ref{0.05} provide the empirical power comparison at significance level $\alpha'=0.01$ and $\alpha'=0.05$, respectively. From Table \ref{0.01} and Table \ref{0.05}, it is clear that the proposed test is better than other competent tests.
\section{Data Analysis}
We use two real data sets to demonstrate the developed testing procedure. The parametric bootstrap procedure given below is used to find the p--values.
\begin{enumerate}
    \item From the observed data, first compute the maximum likelihood estimate of the Rayleigh distribution parameter.
    \item  Generate 10000 bootstrap samples using the Rayleigh distribution parameter estimated in the previous step,
    \item For each sample simulated in Step 2, compute the value of the test statistic.
    \item The bootstrap p--value is obtained as the proportion of the test statistic value from the previous step, which is greater than the test statistic computed form the original data.
\end{enumerate}
To demonstrate the proposed test, we first analyze the real data which give the failure times of 23 ball bearings, mentioned in Caroni (2002) and the failure times are 17.88, 28.92, 33.00, 41.52, 42.12, 45.60, 48.48, 51.84, 51.96, 54.12, 55.56, 67.80, 67.80, 67.80, 68.88, 84.12, 93.12, 98.64, 105.12, 105.84, 127.92, 128.04, 173.40. The data represents the number of million revolutions before failure for each 23 ball bearings ordered according to life endurance. Lieblein and Zelen (1956) used the dataset originally. For studying Rayleigh distribution, the data were analyzed by many authors comprising Kim and Han (2009), Dey and Dey (2014), Vaisakh et al. (2023) and Smitha et al. (2023). The scale parameter of the Rayleigh distribution is estimated as $\widehat\sigma=57.23062$. The bootstrap p--value is obtained as 0.4921 using the above mentioned algorithm. Therefore, we accept the null hypothesis that the data follow a Rayleigh distribution. It is clear that the results of previous studies match this result. The Q--Q plot in figure 1 also shows that the data follows Rayleigh distribution.
\begin{figure}
    \centering
    \includegraphics[width=0.9\linewidth]{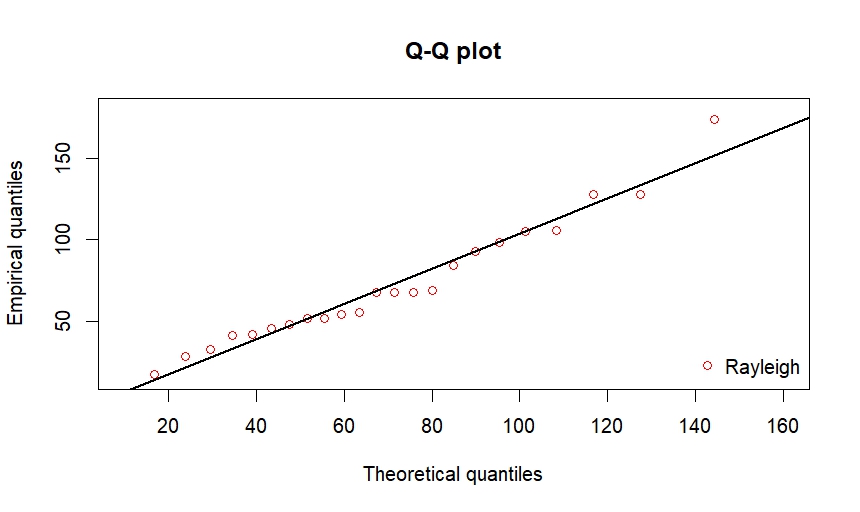}
    \caption{Q--Q plot for ball bearing data}.
    \label{fig:graph 1}
\end{figure}

Next, we analyze the data for survival times in weeks for 20 male rats that were exposed to high level radiation. The data are due to Furth et al. (1959) and have been discussed by Lawless (2003) and the survival times are 152, 152, 115, 109, 137, 88, 94, 77, 160, 165, 125, 40, 128, 123, 136, 101, 62, 153, 83, 69. The bootstrap p--value is obtained as 0.0179 using the above mentioned algorithm. Therefore, we reject the null hypothesis that the data follow a Rayleigh distribution at a level of significance of 5$\%$. From  Vaisakh et al. (2023), we can see that this data follow a gamma distribution. Hence, it is clear that the proposed test is effective in the goodness of fit test.
\section{Conclusion}

In this paper, we introduced the weighted cumulative residual Mathai--Haubold entropy and studied its properties. Also, its dynamic version named dynamic weighted cumulative residual Mathai--Haubold entropy is developed and studied its behavior under linear transformations, bounds and explicit expressions for some lifetime distributions. Based on the associated measure, some characterization results are obtained and formulate two new classes of life distributions. Then we propose a goodness-of-fit test for the Rayleigh distribution and its performance is evaluated through a Monte Carlo simulation study. The practical applicability of the proposed methodology is demonstrated through two real data sets  .

\section{References}

\begin{enumerate}
\item Anija, C. R., Smitha, S., \& Kattumannil, S. K. (2025).
The cumulative residual Mathai--Haubold entropy and its non-parametric inference.\\
\textit{arXiv}. https://doi.org/10.48550/arXiv.2512.10997.
\item Balakrishnan, N., Buono, F., \& Longobardi, M. (2022). On weighted extropies. \textit{Communications in Statistics – Theory and Methods}, 51(18), 6250--6267.
\item Baratpour, S., \& Khodadadi, F. (2013). A cumulative residual entropy characterization of the Rayleigh distribution and related goodness-of-ft test. \textit{Journal of Statistical Research of Iran}, 9(2), 115--131.
\item Belis, M., \& Guiasu, S. (1968). A quantitative-qualitative measure of information in cybernetic systems (Corresp.). \textit{IEEE Transactions on Information Theory}, 14(4), 593--594.
\item Caroni, C. (2002). The correct ball bearings data. \textit{Lifetime Data Analysis}, 8, 395--399.
\item Chakraborty, S., \& Pradhan, B. (2023). Generalized weighted survival and failure entropies and their dynamic versions. \textit{Communications in Statistics – Theory and Methods}, 52(3), 730--750.
\item Dey, S. \& Dey, T. (2014). Statistical inference for the Rayleigh distribution under progressively
Type-II censoring with binomial removal. \textit{Applied Mathematical Modelling}, 38, 974--982.
\item Dar, J. G. \& Al-Zahrani, B. (2013). On some characterization results of lifetime distributions using Mathai--Haubold residual entropy. \textit{IOSR Journal of Mathematics}, 5(4), 56--60.
\item Ebrahimi, N. (1996). How to measure uncertainty in  the residual lifetime distribution. \textit{Sankhya: The Indian Journal of Statistics, Series A}, 58(1), 48--56.
\item Finkelstein, J. M., \& Schafer, R. E. (1971).
Improved goodness-of-fit tests. \textit{Biometrika}, 58(3), 641--645.
\item Furth, J., Upton, A. C., \& Kimball, A. W. (1959). Late pathologic effects of atomic detonation and their pathogenesis. \textit{Radiation Research Supplement}, 1, 243--264.
\item Jahanshahi, S. M. A., Rad, A. H., \& Fakoor, V. (2016). A goodness-of-ft test for Rayleigh distribution based on Hellinger distance. \textit{Annals of Data Science}, 3(4), 401--411.
\item Khammar, A. H., \& Jahanshahi, S. M. A. (2018). On weighted cumulative residual Tsallis entropy and its dynamic version. \textit{Physica A: Statistical Mechanics and its Applications}, 491, 678--692.
\item Kim, C. \& Han, K. (2009). Estimation of the scale parameter of the Rayleigh distribution with
multiply type–II censored sample. \textit{Journal of Statistical Computation and Simulation}, 79, 965--976.
\item Lawless, J.F. (2003) \textit{Statistical Models and Methods for Lifetime Data}, John Wiley and Sons, New York.
\item Lieblein, J. \& Zelen, M. (1956). Statistical investigation of the fatigue life of deep-groove ball bearings. \textit{Journal of Research of the National Bureau of Standards}, 57, 273--316.
\item    Mathai, A. M., \& Haubold, H. J. (2006). Pathway model, superstatistics, Tsallis statistics, and a generalized measure of entropy. \textit{Physica A:Statistical Mechanics and its Applications}, 375(1), 110--122.
\item Mirali, M., Baratpour, S., \& Fakoor, V. (2017). On weighted cumulative residual entropy. \textit{Communications in Statistics – Theory and Methods}, 46(6), 2857--2869.

\item Rao, M., Chen, Y., Vemuri, B.C., \& Wang, F. (2004). Cumulative residual entropy: A new measure of information. \textit{IEEE Transactions on Information Theory}, 50(6), 1220--1228.
\item   Shannon, C. E. (1948). A mathematical theory of communication. \textit{Bell System Technical Journal}, 27(3), 379--423.
\item Smitha, S., Kattumannil, S. K., \& Sreedevi, E. P. (2024). Weighted cumulative residual entropy generating function and its properties. \textit{arXiv. https://arxiv.org/abs/2402.06571}.

\item Vaisakh, K. M., Xavier, T. \& Sreedevi, E. P. (2023). Goodness of fit test for Rayleigh distribution
with censored observations. \textit{Journal of the Korean Statistical Society}, 52, 794--815.

\end{enumerate}
\end{document}